\documentclass[11pt,a4paper]{article}

\usepackage[truedimen,margin=34mm]{geometry} 

\usepackage{mathrsfs}
\usepackage{amssymb}
\usepackage{amsmath}
\usepackage{ascmac}
\usepackage{amsthm}
\usepackage[pdftex]{graphicx}
\usepackage[pdftex]{color}
\usepackage{booktabs}
\usepackage{setspace}
\usepackage{natbib}
\usepackage{color}
\usepackage{url}

\usepackage{titlesec}
\titleformat*{\section}{\large\bfseries}
\titleformat*{\subsection}{\it}

\newtheorem{thm}{Theorem}
\newtheorem{lem}{Lemma}

\newtheorem{algo}{Algorithm}

\def\ep{{\varepsilon}}

\def\beh{{\widehat \beta}}

\def\pih{{\widehat \pi}}

\def\psih{{\widehat \psi}}
\def\muh{{\widehat \mu}}

\def\Thh{{\widehat \Theta}}

\def\fh{\widehat{f}}
\def\gh{\widehat{g}}

\def\ft{\widetilde{f}}

\def\Gh{\widehat{G}}
\def\Lh{\widehat{L}}
\def\Qh{\widehat{Q}}
\def\Qt{\widetilde{Q}}
\def\psit{\widetilde{\psi}}
\def\pit{\widetilde{\pi}}

\def\Re{\mathbb{R}}
\def\E{{\rm E}}
\def\P{{\rm P}}
\def\Ind{{\rm I}}

\title{{\bf Grouped Heterogeneous Mixture Modeling for Clustered Data }}
\date{}

\begin{document}
\doublespacing
\maketitle

\vspace{-2cm}
\begin{center}
{\large Shonosuke Sugasawa\\
{\it Center for Spatial Information Science, The University of Tokyo}
}
\end{center}

\vspace{0.2cm}
\begin{center}
{\large\bf Abstract}
\end{center}
Clustered data is ubiquitous in a variety of scientific fields. In this paper, we propose a flexible and interpretable modeling approach, called grouped heterogenous mixture modeling, for clustered data, which models cluster-wise conditional distributions by mixtures of latent conditional distributions common to all the clusters. In the model, we assume that clusters are divided into a finite number of groups and mixing proportions are the same within the same group. We provide a simple generalized EM algorithm for computing the maximum likelihood estimator, and an information criterion to select the numbers of groups and latent distributions. We also propose structured grouping strategies by introducing penalties on grouping parameters in the likelihood function. Under the settings where both the number of clusters and cluster sizes tend to infinity, we present asymptotic properties of the maximum likelihood estimator and the information criterion. We demonstrate the proposed method through simulation studies and an application to crime risk modeling in Tokyo.

\bigskip\noindent
{\bf Key words}: 
EM algorithm; Finite mixture; Maximum likelihood estimation; Mixture of experts; Unobserved heterogeneity

\newpage
\section{Introduction}\label{sec:intro}
Data having clustered structures such as postal area, school, individual and species is called clustered data, which appears in a variety of scientific fields.
The main goal of statistical analysis of clustered data would be modeling the response variable as a function of covariates or auxiliary variables while accounting for heterogeneity among clusters. 
Although mixed models or generalized linear mixed models \citep[e.g.][]{Dem2013,Jiang2007} have been adopted as a powerful tool to analyze such data, fitting mixed models under non-normal distributions is computationally challenging and they are not necessarily suitable when the conditional distribution is skewed or multimodal. 
As an alternative strategy, finite mixture modeling \citep{Mc2000} could be useful to flexibly capture distributional relationship between response variables and covariates.
For independent (non-clustered) data, the mixture model with covariates was originally proposed in \cite{Jacobs1991}, known as mixture-of-experts, and a large body of literature has been concerned with flexible modeling of the conditional distribution of non-clustered data \citep[e.g.][]{CD2009, GK2007, Jordan1994, Nguyen2016, Villani2009,Villani2012}.
However, these methods cannot be directly imported into analyzing clustered data since they cannot take account of heterogeneity among clusters.

For modeling cluster-wise conditional distributions, \cite{Rubin1997} and \cite{Ng2014} proposed mixtures of random effects models, but this approach would be computationally intensive since the estimation of a single random effects model is already challenging under non-normal responses.
Moreover, results of such methods are not necessarily interpretable due to the large number of random effects to express between-cluster heterogeneity. 
Alternatively, the cluster-wise conditional distributions can be modeled as the following mixture:
\begin{equation}\label{LMR}
f_i(y|x)=\sum_{k=1}^L\pi_{ik}h_k(y|x;\phi_k), \ \ \ \ i=1,\ldots,m,
\end{equation}
where $m$ is the number of clusters, $h_k(y|x;\phi_k)$ is a latent conditional distribution with parameter $\phi_k$.
In the model (\ref{LMR}), the cluster-wise distributions share the common latent distributions $h_k$, and between-cluster heterogeneity is captured only through the cluster-wise mixing proportions $(\pi_{i1},\ldots,\pi_{iL})$ which play a similar role to random effects.
This approach was considered in \cite{Sun2007} and \cite{Sugasawa2019} which employ logistic mixed models and Dirichlet distributions for modeling the mixing proportions, respectively. 
However, both methods still require complicated numerical algorithms to estimate model parameters.  
Apart from parametric modeling of cluster-wise conditional distributions, \cite{Rosen2000} and \cite{Tang2016} proposed a mixture modeling based on the generalized estimating equation, but the primary interest in these works is estimation of the component distributions in the mixture by taking account of correlations within clusters rather than between-cluster heterogeneity.
Moreover, as nonparametric Bayesian approaches, \cite{Teh2006} and \cite{Rod2008} might have similar philosophy to the mixture modeling (\ref{LMR}) in that they assume sharing mixing components among clusters, while they require Markov Chain Monte Carlo algorithm to fit the models.
Therefore, existing approaches for modeling cluster-wise conditional distributions are computationally intensive and may have poor interpretability, so that effective approaches to estimating cluster-wise conditional distributions while accounting for between-cluster heterogeneity are still limited.

In this paper, we seek an effective mixture modeling that is flexible, interpretable and computationally tractable. 
We work on the model formulation (\ref{LMR}) and propose a grouping strategy for cluster-wise mixing proportions, that is, we assume that all the clusters are divided into a finite number of groups and mixing proportions within the same groups share the same mixing proportions, which we call grouped heterogenous mixture (GHM) modeling.
We introduce unknown group membership parameters, and estimate them together with unknown mixing proportions and parameters in latent distributions via the maximum likelihood method. 
Owing to such a simple structure of the mixing proportion, the maximum likelihood estimator can be easily obtained via the generalized EM algorithm \citep{Meng1993}, in which all the steps are simple and do not require any computationally intensive methods such as numerical integrations.
We propose selecting the numbers of groups and latent distributions using an information criterion.
Under the framework that both the number of clusters and within-cluster sample sizes go to infinity, we show that the grouping assignment and unknown model parameters can be consistently estimated. 
We also show that the numbers of groups and latent distributions can be consistently estimated via the information criterion under the same asymptotic framework.

The estimation of group membership parameters is closely related to {\it k-means} algorithm \citep{Forgy1965}, and the idea of grouping parameters have been recently considered in the context of panel data analysis \citep[e.g.][]{Ando2016, BM2015, HM2010}, but this paper is the first one to adopt the idea for modeling cluster-wise mixing proportions. 
Moreover, in this paper, we also propose a new approach to estimating grouping parameters, referred to structured grouping, which incorporates auxiliary information of clusters into grouping assignment to obtain more interpretable results. 
We realize this strategy via simple penalization on grouping parameters, motivated from stochastic processes on discrete spaces, and demonstrate this strategy with spatial information in our application given in Section \ref{sec:crime}.

The rest of paper is organized as follows.
In Section \ref{sec:method}, we introduce the proposed model, describe the generalized EM algorithm, the information criterion for selecting numbers of groups and latent distributions, and then discuss structured grouping via penalization.
In Section \ref{sec:asymp}, we illustrate the asymptotic properties of the maximum likelihood estimator and the information criterion. 
In Section \ref{sec:num}, we carry out simulation studies to demonstrate the proposed method compared with some existing methods, and apply the proposed method to crime risk modeling in 23 wards of Tokyo by adopting mixture models with zero-inflated Poisson regression models.
Finally, some discussions are given in Section \ref{sec:disc}.
All the proofs are given in Appendix, and R code implementing the proposed method is available at Github repository (\url{https://github.com/sshonosuke/GHM}).

\section{Grouped Heterogeneous Mixture Modeling}\label{sec:method}

\subsection{The proposed model}\label{sec:model}
Let $y_{ij}$ denote the $j$th measurement in the $i$th cluster for $i=1,\ldots,m$ and $j=1,\ldots,n_i$, where $m$ is the number of clusters and $n_i$ is the number of within-cluster samples or cluster size which could be unequal over clusters.
Let $x_{ij}$ be a $p$-dimensional vector of covariates associated with $y_{ij}$, where we assume $p$ is fixed.
We are interested in the conditional density or probability mass function of $y_{ij}$ given $x_{ij}$, denoted by $f_i(y_{ij}|x_{ij})$, which could be heterogeneous across clusters.
For modeling $f_i$, we propose the following grouped heterogeneous mixture (GHM) model:
\begin{equation}\label{model}
f_i(y_{ij}|x_{ij})=\sum_{k=1}^L\pi_{g_ik}h_k(y_{ij}|x_{ij};\phi_k), \ \ \ \ i=1,\ldots,m, \ \ \ j=1,\ldots,n_i,
\end{equation}
where $h_k$'s are latent conditional distributions common to all the clusters, $g_i\in \{1,\ldots,G\}$ denotes an unknown group membership parameter and $\phi_k$ is a vector of unknown parameters in the $k$th latent distributions.
As mentioned in Section \ref{sec:intro}, the important feature of the GHM model (\ref{model}) is the structure that $m$ clusters are divided into $G$ groups, and clusters within the same groups share the same mixing proportions, which results in the same conditional distributions in the same groups.
Therefore, the GHM model would be easy to interpret while it is still flexible to account for the between-cluster heterogeneity.
The latent conditional distribution $h_k$ could be specified depending on the type of response variable $y_{ij}$.
For instance, we may use a normal distribution with mean $x^t\beta_k$ and variance $\sigma_k^2$ when the response variable is continuous, and a Poisson distribution with rate parameter $\exp(x^t\beta_k)$ when the response variable is count valued.
Note that all the $h_k$ does not have to be the same family.
In fact, we adopt combinations of a zero-inflation component and Poisson regression models as latent distributions in the application to crime data given in Section \ref{sec:crime}.

\subsection{Generalized EM algorithm}\label{sec:EM}
We estimate the grouping parameter $g_i$ and unknown parameters $\phi_k$ and $\pi_{gk}$ simultaneously via the maximum likelihood method. 
Define $\Theta=(\gamma,\pi,\phi)$ be the set of unknown parameters, where $\gamma, \pi$ and $\phi$ are collections of $g_i, \pi_{gk}$ and $\phi_k$, respectively. 
We here assume that the numbers of groups $G$ and latent distributions $L$ are known, but the estimation will be discussed later.
Given the data, the log-likelihood function of $\Theta$ is given by
\begin{equation*}
Q(\Theta)=\sum_{i=1}^m\sum_{j=1}^{n_i}\log\left(\sum_{k=1}^L\pi_{g_ik}h_k(y_{ij}|x_{ij};\phi_k)\right).
\end{equation*}
The maximum likelihood estimator $\Thh$ of $\Theta$ is defined as the maximizer of $Q(\Theta)$.
We develop an iterative method for maximizing $Q(\Theta)$ using the generalized EM algorithm \citep{Meng1993}.
To this end, we introduce latent variables $z_{ij}\in\{1,\ldots,L\}$ and consider a hierarchical expression of the model (\ref{model}) given by $y_{ij}|(z_{ij}=k)\sim h_k(y_{ij}|x_{ij};\phi_k)$ and $\P(z_{ij}=k)=\pi_{g_ik}$.
Given $z_{ij}$'s, the complete log-likelihood is given by
$$
L^c(\Theta)=
\sum_{i=1}^m\sum_{j=1}^{n_i}\sum_{k=1}^{L}
\Ind(z_{ij}=k)\left\{\log h_k(y_{ij}|x_{ij};\phi_k)+\log \pi_{g_ik}\right\}.
$$
In the E-step, we compute the conditional probabilities being $z_{ij}=k$ for $k=1,\ldots,L$, given the data and current parameter values $\Theta^{(r)}$, which are denoted by $w_{ijk}^{(r)}$.
The detailed expression of $w_{ijk}^{(r)}$ will be given in Algorithm \ref{algo:GEM}.   
Then, the objective function in the M-step, which is the expectation of the complete log-likelihood, is given by
\begin{equation}\label{Q}
\begin{split}
Q(\Theta|\Theta^{(r)})
&=\sum_{k=1}^{L}\sum_{i=1}^m\sum_{j=1}^{n_i}w_{ijk}^{(r)}\log h_k(y_{ij}|x_{ij};\phi_k)
+\sum_{i=1}^m\sum_{k=1}^{L}\log \pi_{g_ik}\sum_{j=1}^{n_i}w_{ijk}^{(r)}\\
&\equiv Q_1(\phi|\Theta^{(r)})+Q_2(\gamma,\pi|\Theta^{(r)}).
\end{split}
\end{equation}
It is observed that the maximization of $Q_1(\phi|\Theta^{(r)})$ with respect to $\phi$ can be divided into $L$ maximization problems and $\phi_1,\ldots,\phi_L$ can be separately updated.
On the other hand, maximizing $Q_2(\gamma,\pi|\Theta^{(r)})$ includes discrete optimization of $\gamma$ on the space $\{1,\ldots,G\}^m$, and brute-force search over the space would be infeasible. 
Instead of simultaneous maximization, we first maximize $Q_2(\gamma^{(r)},\pi|\Theta^{(r)})$ with respect to $\pi$ and we set $\pi^{(r+1)}$ to the maximizer, and then, maximize $Q_2(\gamma,\pi^{(r+1)}|\Theta^{(r)})$ with respect to $\gamma$ to get $\gamma^{(r+1)}$, which allows separate updating for each element of $\gamma$.
This updating process guarantees that monotone increasing of the objective function, that is, $Q_2(\gamma^{(r)},\pi^{(r)}|\Theta^{(r)})\leq Q_2(\gamma^{(r+1)},\pi^{(r+1)}|\Theta^{(r)})$.
The proposed GEM algorithm is summarized in what follows.

\vspace{0.2cm}
\begin{algo}[Generalized EM algorithm] \ \ \   \label{algo:GEM}
Starting from the initial values $\Theta^{(0)}$ and $r=0$, we repeat the following procedure until the algorithm converges:
\begin{itemize}
\item
(E-step) \ \ 
Compute the following weights:
$$ 
w_{ijk}^{(r)}=
\frac{\pi^{(r)}_{g_ik}h_k(y_{ij}|x_{ij};\phi_k^{(r)})}{\sum_{\ell=1}^L\pi^{(r)}_{g_i\ell}h_\ell(y_{ij}|x_{ij};\phi_\ell^{(r)})},
$$
for $j=1,\ldots,n_i$ and $i=1,\ldots,m$.

\item
(M-step) \ \ 
Update unknown parameters as follows: 
\begin{align*}
&\phi_k^{(r+1)}
={\rm argmax} \ \bigg\{\sum_{i=1}^m\sum_{j=1}^{n_i}w_{ijk}^{(r)}\log h_k(y_{ij}|x_{ij};\phi_k)\bigg\}, \\
&\pi_{gk}^{(r+1)}
=\bigg(\sum_{}{}_{i:g_i^{(r)}=g}n_i\bigg)^{-1}
\sum_{}{}_{i:g_i^{(r)}=g}\sum_{j=1}^{n_i}w_{ijk}^{(r)}, \\
&g_i^{(r+1)}={\rm argmax}_{g=1,\ldots,G} \ \bigg(\sum_{k=1}^{L}\log\pi_{gk}^{(r+1)}\sum_{j=1}^{n_i}w_{ijk}^{(r)}\bigg),  
\end{align*}
for $i=1,\ldots,m$, $g=1,\ldots,G$ and $k=1,\ldots,L$.
\end{itemize}
\end{algo}

\bigskip
Note that updating $\phi_k$ requires maximizing the weighted log-likelihood of $k$th latent distribution, which can be readily carried out as long as the latent density $h_k$ is familiar regression models such as normal linear regression or Poisson regression.
Moreover, for updating $g_i$, we simply compute all the values of the objective function for $g=1,\ldots,G$ and select the maximizer.
The above Algorithm \ref{algo:GEM} can be seen as a combination of the classical EM algorithm for mixture models and a modified {\it k-means} algorithm.

The above algorithm could be sensitive to the choice of initial estimators of $\Theta$, and we propose the following strategy to obtain reasonable initial estimators. For each cluster, we firstly fit different finite mixture models with $L$ components to get estimates of cluster-wise mixing proportions $\pi_{i\ast}^{(0)}=(\pi_{i1}^{(0)},\ldots,\pi_{i,L-1}^{(0)})$ as well as parameters $\phi_{ik}^{(0)}$ in the mixing distributions.
Then, we apply the standard {\it k-means} algorithm with $G$ groups to $\pi_{1\ast}^{(0)},\ldots,\pi_{m\ast}^{(0)}$ and adopt the estimated grouping and centroid as the initial estimators of $\gamma$ and $\pi$, respectively. 
The initial estimators of $\phi_k$ could be obtained via element-wise means or medians of $\{\phi_{1k}^{(0)},\ldots,\phi_{mk}^{(0)}\}$ for $k=1,\ldots,L$.

\subsection{Selection of the numbers of groups and latent distributions}
The proposed mixture model has two tuning parameters, $G$ (the number of groups) and $L$ (the number of latent distributions), which would be critically related to performance and interpretability of the proposed method. 
Using unnecessarily small values of $G$ or $L$ would result in failure to capture the between-cluster heterogeneity or loss of flexibility for modeling the conditional distributions, whereas the use of unnecessarily large values of $G$ or $L$ may lead to redundant modeling which does not hold meaningful interpretations.    
Therefore, we suggest selecting adequate values of $G$ and $L$ by minimizing the following information criterion:
\begin{equation}\label{BIC}
{\rm IC}(G,L)=-2Q(\widehat{\Theta}_{G, L})+\log N\left\{G(L-1)+m+\sum_{k=1}^L{\dim}(\phi_k)\right\},
\end{equation} 
where $N=\sum_{i=1}^mn_i$ is the total sample size, $\Theta_{G, L}$ is the set of unknown parameters under $G$ groups and $L$ latent distributions, and ${\dim}(\phi_k)$ is the dimension of the component-specific parameter $\phi_k$.
Clearly, the above information criterion is motivated from Bayesian information criterion \citep{BIC}, which is known to provide consistent estimation of the number of components in the standard finite mixture models \citep[e.g.][]{Keribin2000}. 
In Section \ref{sec:asymp}, we show that the information criterion (\ref{BIC}) enables us to obtain consistent estimators of $L$ as well as $G$. 
The minimizer can be easily searched over all the combinations of pre-specified candidates of $G$ and $L$, namely $G\in\{G_{\rm min}, \ldots, G_{\rm max}\}$ and $L\in\{L_{\rm min}, \ldots, L_{\rm max}\}$. 
Typically, we adopt $G_{\rm min}=1$ and $L_{\rm min}=1$, and set $G_{\rm max}$ and $L_{\rm min}$ to large and moderate values, respectively.
Since the computation of $\widehat{\Theta}_{G, L}$ for each choice of $(G, L)$ can be easily computed via Algorithm \ref{algo:GEM}, the brute-force search would be feasible in practice.

\subsection{Structured grouping}\label{sec:pen}
In the proposed method, the grouping membership is determined in a model-dependent way. 
On the other hand, we may also incorporate auxiliary or external cluster-level information into grouping. 
Let $s_{ii'}$ be some similarity measure between clusters $i$ and $i'$, satisfying $s_{ii'}\in [0,1]$. 
For example, if the cluster-level covariate $z_i$ is available, we may define $s_{ii'}=\exp(-b\|z_i-z_{i'}\|^2)$ with a scaling constant $b$.
In our application in Section \ref{sec:crime}, we consider incorporating adjacent information of areas (clusters) into grouping assignment.
Specifically, we adopt the indicator of adjacency between $i$th and $i'$th areas for $s_{ii'}$, that is, $s_{ii'}=1$ if two areas are adjacent and $s_{ii'}=0$ otherwise.    
Under the settings, it would be reasonable to make two clusters having larger similarity $s_{ii'}$ be more likely to be classified in the same group, which would make the results more interpretable.
To incorporate such information in the estimation process, we consider maximizing the following penalized likelihood function: 
\begin{equation}\label{PML}
Q(\Theta)+\xi\sum_{i<i'}s_{ii'}{\rm I}(g_i=g_{i'}),
\end{equation}   
where $\xi\geq 0$ is a tuning parameter that controls the strength of the penalty.
The penalty term in (\ref{PML}) is motivated from a stochastic process on the space of grouping parameters, i.e. $\{1,\ldots,G\}^{m}$, called Potts model \citep{Potts1952}, where its probability function is proportional to $\exp(\xi\sum_{i<i'}s_{ii'}{\rm I}(g_i=g_{i'}))$.
Since $Q(\Theta)$ is the log-likelihood function, maximizing the penalized likelihood function (\ref{PML}) can be seen as the maximum a posterior probability estimation under the prior stochastic process.
The maximization of (\ref{PML}) can be still easily carried out by a similar algorithm to Algorithm \ref{algo:GEM}, where the details are given as follows:

\vspace{0.2cm}
\begin{algo}[Modified generalized EM algorithm under penalization] \ \  \label{algo:GEM2}
The E-step and M-steps for $\phi_k$ and $\pi_{gk}$ are the same as those in Algorithm \ref{algo:GEM}, but the M-step for $g_i$ is given by
\begin{align*}
&g_i^{(r+1)}={\rm argmax}_{g=1,\ldots,G} \ \bigg(\sum_{k=1}^{L}\log\pi_{gk}^{(r+1)}\sum_{j=1}^{n_i}w_{ijk}^{(r)} + \xi\sum_{i'=1}^{m}s_{ii'}{\rm I}(g=g_{i'}^{(r)})\bigg),  
\end{align*}
where $i=1,\ldots,m$.
\end{algo}

\medskip
We note that the updating process for $g_i$ can be separately carried out for $i=1,\ldots,m$, thereby we can simply calculate the values of the objective function for $g=1,\ldots,G$, and choose the maximizer.   
Regarding the value of $\xi$, we here recommend simply using small values such as $\xi=0.2$ or $\xi=0.5$, as adopted in Section \ref{sec:crime}.
If we adopt a too large value of $\xi$, the penalty term will be too strong and the information from the likelihood would be ignored in the estimation of $g_i$, which is quite unreasonable. 
Although it might be able to select $\xi$ via some data-dependent manner, we do not consider the details further since it would extend the scope of this paper.  
We also note that all the theoretical results given in Section \ref{sec:asymp} would be valid for the penalized likelihood (\ref{PML}) as long as the penalty term is asymptotically negligible.

\section{Asymptotic Properties}\label{sec:asymp}
We investigate the large-sample properties of the maximum likelihood estimator $\Thh$ and selection performance of the information criterion (\ref{BIC}) when both the number of clusters and cluster sizes tend to infinity. 
Without loss of generality, we suppose that the cluster labels are assigned such that the cluster size increases with the cluster index, that is, $n_1\leq n_2\leq \cdots \leq n_m$.
We assume that $m\to\infty$ as well as $n_1\to\infty$ as considered in literatures \citep[e.g.][]{Vonesh2002,Hui2017b}.
We also adopt the setting that $m/n_1^{\nu}\to 0$ for some $\nu>0$, that is, the cluster sizes do not have to grow at the same rate as $m$, and assume that $n_m=O(n_1^{\alpha})$ for some $\alpha>1$, that is, the largest cluster size $n_m$ is the polynomial order of the smallest cluster size $n_1$. 
In Theorem \ref{thm:main}, we further assume that $G$ and $L$ are known.
We then obtain the following asymptotic results for the maximum likelihood estimator.

\begin{thm}\label{thm:main}
Under regularity conditions given in Appendix and $m, n_1\to\infty$ and $m/n_1^{\nu}\to 0$ for some $\nu>0$, it holds that 
\begin{align}
&\frac1m\sum_{i=1}^m\sum_{k=1}^L(\pih_{\gh_ik}-\pi^0_{g_i^0k})^2=O_p\Big(\frac1m\Big)  \label{st1}\\
&\sqrt{N}F_m(\Thh)^{1/2}(\psih-\psi^0)\to N(0,I_{{\rm dim}(\psi)}),  \label{st2}
\end{align}
where $F_m(\Theta)$ is the Fisher information matrix defined in Appendix, and $g_i^0$, $\pi_{gk}^{0}$ and $\psi^0$ are the true parameters.
\end{thm}
\vspace{0.3cm}

Although the number of cluster-wise mixing proportions grows at the number of clusters $m$, (\ref{st1}) shows that they are consistently estimated and the convergence rate depends on $m$.
Moreover, (\ref{st2}) shows that the structural parameters $\phi$ in latent distributions and group-wise mixing proportions $\pi$ whose dimensions are fixed are $\sqrt{N}$-consistent.
We next consider asymptotic properties of the information criterion (\ref{BIC}) to select $G$ and $L$.
We obtain the following result.

\begin{thm}\label{thm:selection}
Let $(\Gh,\Lh)$ be the minimizer of the information criterion (\ref{BIC}), and $(G^0, L^0)$ be the true values. 
Under regularity conditions given in Appendix and $m\to\infty$ and $m/n_1^{\nu}\to 0$ for some $\nu>0$, it holds that $P\{(\Gh, \Lh)=(G^0,L^0)\}\to 1$. 
\end{thm}

Theorem \ref{thm:selection} shows that the information criterion (\ref{BIC}) holds selection consistency for both $G$ and $L$.
This can be seen as an extension of the results regarding selection consistency of BIC in the standard mixture model \citep{Keribin2000}.

\section{Numerical Studies}\label{sec:num}

\subsection{Simulation study: continuous response}\label{sec:sim-G}
We investigate the finite sample performance of the proposed method together with some existing methods.
For simplicity, we set the cluster sizes to be the same in all the clusters, that is, $n_i=n$ for $i=1,\ldots,m$, where $m$ is the number of clusters.
We first generated two covariates $x_{1ij}$ and $x_{2ij}$ from the uniform distribution on $(-0.2, 0.8)$ and Bernoulli distribution with success probability 0.5, respectively. 
We considered five scenarios for the true structures of cluster-wise conditional distributions written in the following form:
\begin{align*}
f_i(y_{ij}|x_{1ij},x_{2ij})
&=\pi_i\phi(y_{ij}; \beta_{i10}+\beta_{i11}x_{1ij}+\beta_{i12}x_{2ij}, \sigma_1^2)\\
& \ \ \ \ +(1-\pi_i)\phi(y_{ij}; \beta_{i20}+\beta_{i21}x_{1ij}+\beta_{i22}x_{2ij}, \sigma_2^2).
\end{align*}
where $\phi(\cdot; a,b)$ denotes a normal density function with mean $a$ and variance $b$.
In the first three scenarios, we set $\beta_{i20}=\beta_{i22}=-\beta_{i10}=-\beta_{i12}=0.5$, $\beta_{i11}=-\beta_{i21}=1$ for all $i$, and set $\sigma_1=0.2$ and $\sigma_2=0.5$, thereby all the clusters share the same latent distributions and variations of cluster-wise conditional distributions are determined by the cluster-specific mixing proportion $\pi_i$.
The generating processes of $\pi_i$ are given by ${\rm (I)}\ {\rm Beta}(2,1)$, ${\rm (II)}\ {\rm TP}(0.1, 0.9)$ and ${\rm (III)}\ {\rm TP}(0.1, 0.9)+U(-0.1,0.1)$, where TP$(a,b)$ denotes a two point distribution on $a$ and $b$ with equal probabilities.  
We note that scenario (I) produces different mixing proportions over all the clusters, while $\pi_i$ takes only two values in scenario (II), that is, all the clusters can be divided into two groups.
Scenario (III) corresponds to the case where the clusters can be approximately clustered into two groups.
In scenario (IV), we generated the regression coefficients from $\beta_{i10}, \beta_{i12}\sim N(-0.5, (0.3)^2)$, $\beta_{i20}, \beta_{i22}\sim N(0.5, (0.3)^2)$, $\beta_{i11}\sim N(1, (0.3)^2)$ and $ \beta_{i21}\sim N(-1, (0.3)^2)$ and set $\sigma_1=0.2$ and $\sigma_2=0.5$.
We also generated the mixing proportion $\pi$ from ${\rm Beta}(2,1)$.
In this scenario, both mixing proportions and distributions are different in each cluster.
Finally, we considered a more simplified structure as scenario (V), that is, we set $\pi_i=1$ for all $i$, and adopted $\beta_{i10}, \beta_{i11}, \beta_{i12}\sim N(0,(0.5)^2)$ and $\sigma=0.3$.
Regarding the settings of $m$ (number of clusters) and $n$ (number of within-cluster sample size), we considered five settings, namely, $(m, n)=(20, 80), (20, 160), (40, 80), (40, 160)$ and $(60, 80)$.

For the simulated dataset, we applied the proposed GHM model with normal linear regression models as latent conditional distributions in which $G$ (the numbers of groups) and $L$ (the number of latent distributions) are selected by the information criterion (\ref{BIC}) among combinations of $G\in \{1,2,\ldots,10\}$ and $L\in \{1,2,3,4\}$.
We also applied the proposed method with $L=2$ and the largest number of groups, $G=10$, denoted by fGHM.
For comparison, we firstly consider two logical options that directly apply the standard finite mixture of normal linear regression models, which we call global mixture (GM) and local mixture (LM).
The GM method applies the finite mixture model to all the data by ignoring the nature of clustered data, while the LM method applies the mixture model separately to the clusters. 
The number of components in both GM and LM were selected from $\{1,\ldots,4\}$ by BIC.
As a competitor from random effects models, we adopted random coefficient (RC) models using R package ``lme4" \citep{Bates2016}.
Moreover, as a competitor from an advanced mixture methods for clustered data, we applied the latent Dirichlet mixture (LDM) approach with normal linear regression models as latent conditional distributions, developed in \cite{Sugasawa2019}, which models the cluster-wise mixing proportions via the Dirichlet distribution.
It should be noted that the computational cost of LDM is quite intensive since it requires a Monte Carlo EM algorithm including Markov Chain Monte Carlo algorithm in its E-step, thereby we set the number of component to $2$ in all the cases.

Based on these methods, we computed estimators of cluster-wise conditional distributions $\fh_i$, and compute the following mean integrated squared errors (MISE):
\begin{align}\label{MISE}
{\rm MISE}
&=\frac1{m}\sum_{i=1}^m\sum_{x_2\in \{0,1\}}\int_{-0.2}^{0.8}\int_{-\infty}^{\infty}\Big\{\fh_i(y|x_1,x_2)-f_i(y|x_1,x_2)\Big\}^2{\rm d}y{\rm d}x_1, \notag \\
& \approx \frac1{m}\sum_{i=1}^m\sum_{x_2\in \{0,1\}}\sum_{x_1\in \Delta_{x_1}}\sum_{y\in \Delta_y}b_{x_1}b_y\Big\{\fh_i(y|x_1,x_2)-f_i(y|x_1,x_2)\Big\}^2,
\end{align}
where $\Delta_{x_1}$ and $\Delta_y$ are sets of equally-spaced points from $-0.2$ to $0.8$ with $b_{x_1}$ increment  and from $-5$ to $5$ with $b_y$ increment, respectively, and we set $b_{x_1}=0.02$ and $b_y=0.1$. 
In Table \ref{tab:sim-Gauss}, we reported the MISE values averaged over 300 replications, for five scenarios of data generating mechanism with five combinations of $(m, n)$.
It is observed that the LM and GM approaches do not provide preferable performance compared with the proposed methods possibly due to the lack of model flexibility in GM and instability of estimation in LM under moderate sample sizes such as $n=80$ and $160$.
Although RC works quite well when the model is correctly specified as in scenario (V), its performance is not necessarily satisfactory under the other scenarios. 
On the other hand, the LDM approach works reasonably well in all the scenarios by flexibly modeling the cluster-wise mixing proportions via the Dirichlet distribution. 
However, overall the proposed GHM performs as well as or better than the LDM method. 
Moreover, it is interesting to see that the performance of GHM is quite comparable with LDM even under scenario (I) in which mixing proportions are completely different over clusters, while GHM imposes clustering structures.  
In scenario (IV), both mixing proportions and distributions are different, thereby both GHM and LDM are misspecified, but GHM provides better estimation performance than LDM.
In Table \ref{tab:GHM-Gauss}, we reported average values of selected numbers of $G$ and $L$ in GHM. 
In scenarios (I)$\sim$(III), the average values of selected $L$ are almost equal to the true value $2$, and the average value of $G$ tends to increase in the order of scenarios (II), (III) and (I), as the variations in mixing proportions increase.
We note that the true data generating process is consistent with the proposed GHM in scenario (II), and the average values of the selected $G$ is almost equal to $2$, which would support the theoretical result given in Theorem \ref{thm:selection}. 
In these scenarios, the performance of GHM and fGHM are comparable, so that the effect of selecting smaller numbers of $G$ on estimation performance could be limited.
In scenarios (IV) and (V), the average numbers of $G$ and $L$ are larger than those in the other scenarios possibly because the data generating processes are much more complicated than scenarios (I)$\sim$(III).  
In these cases, selecting suitable values of $G$ as well as $L$ in GHM leads to the improvement of estimation performance compared with fGHM.

\begin{table}
\caption{
Square root of mean integrated squared errors (MISE) averaged over 300 replications for six methods, the proposed grouped heterogeneous mixture (GHM and fGHM), global mixture (GM), local mixture (LM), random coefficient (RC) and latent Dirichlet mixture (LDM), under continuous responses with five scenarios and five combinations of $m$ (the number of clusters) and $n$ (within-cluster sample sizes). The following values are multiplied by 10.
\label{tab:sim-Gauss}
}
\begin{center}
\begin{tabular}{cccccccccccccccc}
\hline
Scenario & $(m ,n)$ &  & GHM & fGHM & GM & LM & RC & LDM \\
 \hline
 & (20, 80) &  & 3.95 & 3.94 & 5.79 & 24.09 & 5.03 & 3.92 \\
 & (20, 160) &  & 3.90 & 3.91 & 5.76 & 6.76 & 5.01 & 3.87 \\
(I) & (40, 80) &  & 3.94 & 3.94 & 5.78 & 22.11 & 5.03 & 3.96 \\
 & (40, 160) &  & 3.90 & 3.91 & 5.76 & 7.89 & 5.01 & 3.89 \\
 & (60, 80) &  & 3.94 & 3.97 & 5.81 & 22.95 & 5.03 & 3.96 \\
 \hline
 & (20, 80) &  & 4.52 & 4.53 & 7.05 & 23.14 & 4.76 & 4.53 \\
 & (20, 160) &  & 4.51 & 4.53 & 7.02 & 6.59 & 4.73 & 4.52 \\
(II) & (40, 80) &  & 4.50 & 4.52 & 7.07 & 23.50 & 4.75 & 4.52 \\
 & (40, 160) &  & 4.50 & 4.52 & 7.09 & 6.42 & 4.72 & 4.52 \\
 & (60, 80) &  & 4.50 & 4.52 & 7.10 & 24.25 & 4.75 & 4.53 \\
 \hline
 & (20, 80) &  & 4.57 & 4.57 & 7.06 & 24.11 & 4.79 & 4.55 \\
 & (20, 160) &  & 4.55 & 4.55 & 7.06 & 6.16 & 4.75 & 4.54 \\
(III) & (40, 80) &  & 4.55 & 4.55 & 7.12 & 23.94 & 4.79 & 4.56 \\
 & (40, 160) &  & 4.54 & 4.54 & 7.11 & 5.92 & 4.76 & 4.52 \\
 & (60, 80) &  & 4.54 & 4.57 & 7.12 & 23.11 & 4.78 & 4.54 \\
 \hline
 & (20, 80) &  & 4.75 & 5.08 & 5.44 & 22.67 & 5.25 & 5.08 \\
 & (20, 160) &  & 4.66 & 5.04 & 5.43 & 7.29 & 5.19 & 5.04 \\
(IV) & (40, 80) &  & 4.81 & 5.11 & 5.41 & 22.55 & 5.24 & 5.10 \\
 & (40, 160) &  & 4.76 & 5.08 & 5.45 & 7.22 & 5.20 & 5.07 \\
 & (60, 80) &  & 4.84 & 5.13 & 5.41 & 22.45 & 5.23 & 5.11 \\
 \hline
 & (20, 80) &  & 5.85 & 6.09 & 7.01 & 23.71 & 5.88 & 6.05 \\
 & (20, 160) &  & 5.80 & 6.14 & 7.05 & 6.00 & 5.80 & 6.10 \\
(V) & (40, 80) &  & 5.76 & 6.17 & 6.99 & 19.68 & 5.79 & 6.10 \\
 & (40, 160) &  & 5.80 & 6.20 & 7.02 & 6.24 & 5.85 & 6.13 \\
 & (60, 80) &  & 5.82 & 6.23 & 7.03 & 21.76 & 5.84 & 6.14 \\
\hline
\end{tabular}
\end{center}
\end{table}

\begin{table}
\caption{
Average values of selected $G$ (the number of groups) and $L$ (the number of latent distributions) in groped heterogeneous mixture (GHM) modeling under continuous responses.
\label{tab:GHM-Gauss}
}
\begin{center}
\begin{tabular}{cccccccccccccccc}
\hline
&& \multicolumn{2}{c}{(I)} & \multicolumn{2}{c}{(II)} & \multicolumn{2}{c}{(III)} & \multicolumn{2}{c}{(IV)} & \multicolumn{2}{c}{(V)} \\
$(m, n)$ &  & $G$ & $L$ & $G$ & $L$ & $G$ & $L$ & $G$ & $L$ & $G$ & $L$\\
\hline
(20, 80) &  & 4.24 & 2.00 & 2.07 & 2.00 & 3.32 & 2.00 & 6.79 & 3.83 & 7.12 & 3.95\\
(20, 160) &  & 5.21 & 2.01 & 2.03 & 2.00 & 5.17 & 2.00 & 8.02 & 3.90 & 7.81 & 3.96\\
(40, 80) &  & 5.49 & 2.00 & 2.27 & 2.00 & 4.58 & 2.01 & 8.21 & 3.90 & 8.07 & 3.95\\
(40, 160) &  & 6.63 & 2.01 & 2.03 & 2.01 & 5.05 & 2.00 & 8.63 & 3.93 & 8.49 & 4.00\\
(60, 80) &  & 6.43 & 2.02 & 2.29 & 2.00 & 4.82 & 2.02 & 8.59 & 3.90 & 8.25 & 3.96\\
\hline
\end{tabular}
\end{center}
\end{table}

\subsection{Simulation study: count response}\label{sec:sim-Po}
We next investigate the performance under count responses by adopting the following structures of cluster-wise conditional distributions: 
\begin{align*}
f_i(y_{ij}|x_{1ij},x_{2ij})
&=\pi_i{\rm Po}(y_{ij}; \exp(\beta_{i10}+\beta_{i11}x_{1ij}+\beta_{i12}x_{2ij}))\\
& \ \ \ \ +(1-\pi_i){\rm Po}(y_{ij}; \exp(\beta_{i20}+\beta_{i21}x_{1ij}+\beta_{i22}x_{2ij})),
\end{align*}
where ${\rm Po}(\cdot; \lambda)$ denotes the probability mass function of Poisson distribution with rate parameter $\lambda$.
We considered the same five scenarios for the mixing proportion $\pi$ as in Section \ref{sec:sim-G}.
Regarding the regression coefficients, we set $\beta_{i20}=\beta_{i22}=-\beta_{i10}=-\beta_{i12}=0.25$, $\beta_{i11}=-\beta_{i21}=0.5$ for all $i$, in scenario (I), (II) and (III), and generated from $\beta_{i10}, \beta_{i12}\sim N(-0.25, (0.3)^2)$, $\beta_{i20}, \beta_{i22}\sim N(0.25, (0.3)^2)$, $\beta_{i11}\sim N(0.5, (0.3)^2)$ and $ \beta_{i21}\sim N(-0.5, (0.3)^2)$ in scenario (IV).
In scenario (V), we employed the same structure as in Scenario \ref{sec:sim-G}.
Also we considered the same five combinations of $(m, n)$.

For the simulated dataset, we applied the GHM, fGHM, GM, LM and LDM methods with Poisson regression models as latent distributions and the other settings equivalent to the previous simulation study in Section \ref{sec:sim-G}.
Moreover, we applied a Poisson mixed model with random intercept (RI), by using R package ``lme4".
The estimated cluster-wise conditional distributions were evaluated by the MISE (\ref{MISE}) with $\Delta_y=\{0, 1, \ldots, 15\}$.
In Table \ref{tab:sim-Po}, we reported the MISE values averaged over 300 replications. 
Since the true generating process for $\pi_i$ is the same as the assumed model in LDM in scenario (I), it is quite reasonable that the LDM performs the best in this scenario. 
Although the proposed GHM and aGHM work slightly worse than LDM, they perform much better than the other approaches.
This is because GHM pursues computational tractability and interpretable results by grouping clusters at the cost of some loss of flexibility compared with LDM. 
In this scenario, fGHM using $G=10$ performs better than GHM, possibly because the true mixing proportions do not admit grouping structures and the large number of groups could be necessarily for modeling such complicated structures.  
The GHM method performs the best in scenario (II) since the mixing proportion $\pi_i$ admits the grouped structure in this case. 
In the other scenarios, the assumed models for $\pi_i$ in both GHM and LDM are misspecified, and the results show that the two methods are quite comparable, thereby the proposed method would be more appealing in that the proposed method is computationally much less intensive and is more interpretable than LDM.
We also note that the two logical approaches, GM and LM, do not perform well in all the scenarios due to the same reasons given in Section \ref{sec:sim-G}.
Although the RI method performs quite well in scenario (V) since the true data generating process is quite close to RI, the performance in the other scenarios is not necessarily reasonable. 
Overall, the performance gets better as $m$ and $n$ increases, but increasing of $n$ would be more influential than that of $m$.  
The overall trend regarding the change of $m$ and $n$ are quite similar to Table \ref{tab:sim-Gauss}.
Finally, in Table \ref{tab:GHM-Po}, we reported the average values of selected $G$ and $L$ in GHM, which shows that the selected $G$ may adaptively change depending on the complexity of the true structure of $\pi_i$. 
In particular, the average numbers of $G$ and $L$ in scenario (IV) are larger than the other scenarios since the cluster-specific distributions in scenario (IV) hold quite complicated between-cluster heterogeneity.

\begin{table}
\caption{
Square root of mean integrated squared errors (MISE) averaged over 300 replications for six methods, the proposed grouped heterogeneous mixture (GHM and fGHM), global mixture (GM), local mixture (LM), random intercept (RI) and latent Dirichlet mixture (LDM), under count responses with five scenarios and five combinations of $m$ (the number of clusters) and $n$ (within-cluster sample sizes). The following values are multiplied by 100.
}
\label{tab:sim-Po}
\begin{center}
\begin{tabular}{cccccccccccccccc}
\hline
Scenario & $(m ,n)$ &  & GHM & fGHM & GM & LM & RI & LDM \\
 \hline
 & (20, 80) &  & 9.77 & 8.65 & 16.11 & 19.89 & 18.69 & 8.17\\
 & (20, 160) &  & 7.36 & 6.42 & 15.55 & 16.88 & 18.08 & 5.92\\
(I) & (40, 80) &  & 9.06 & 8.48 & 15.67 & 19.82 & 18.58 & 7.47\\
 & (40, 160) &  & 6.89 & 6.52 & 15.42 & 17.05 & 18.00 & 5.44\\
 & (60, 80) &  & 8.77 & 8.43 & 15.47 & 19.89 & 18.55 & 7.14\\
 \hline
 & (20, 80) &  & 4.93 & 6.40 & 29.61 & 14.89 & 18.74 & 6.41\\
 & (20, 160) &  & 3.47 & 4.68 & 29.46 & 11.83 & 18.06 & 4.84\\
(II) & (40, 80) &  & 4.20 & 5.89 & 29.94 & 14.97 & 18.88 & 5.83\\
 & (40, 160) &  & 2.68 & 4.07 & 29.80 & 11.75 & 18.19 & 4.38\\
 & (60, 80) &  & 4.17 & 5.70 & 29.98 & 14.92 & 18.94 & 5.58\\
 \hline
 & (20, 80) &  & 6.47 & 6.74 & 29.99 & 15.04 & 18.85 & 6.47\\
 & (20, 160) &  & 5.10 & 5.07 & 29.69 & 12.01 & 18.11 & 4.83\\
(III) & (40, 80) &  & 6.02 & 6.36 & 30.12 & 15.16 & 19.01 & 5.93\\
 & (40, 160) &  & 4.66 & 4.81 & 30.06 & 11.98 & 18.31 & 4.34\\
 & (60, 80) &  & 5.94 & 6.35 & 30.22 & 15.15 & 19.00 & 5.70\\
 \hline
 & (20, 80) &  & 16.32 & 16.44 & 22.74 & 19.57 & 22.96 & 16.73\\
 & (20, 160) &  & 14.54 & 15.91 & 22.62 & 15.67 & 22.51 & 15.58\\
(IV) & (40, 80) &  & 15.97 & 16.70 & 22.70 & 19.77 & 23.02 & 16.38\\
 & (40, 160) &  & 14.51 & 16.19 & 22.84 & 15.76 & 22.73 & 15.32\\
 & (60, 80) &  & 15.78 & 16.90 & 23.00 & 19.95 & 23.00 & 16.31\\
 \hline
 & (20, 80) &  & 11.47 & 10.97 & 13.64 & 15.27 & 9.35 & 11.37\\
 & (20, 160) &  & 10.35 & 9.95 & 13.37 & 12.10 & 8.30 & 10.03\\
(V) & (40, 80) &  & 10.96 & 10.64 & 13.45 & 14.96 & 8.88 & 10.84\\
 & (40, 160) &  & 10.23 & 10.09 & 13.38 & 12.26 & 8.24 & 9.76\\
 & (60, 80) &  & 10.82 & 10.61 & 13.48 & 14.93 & 8.82 & 10.70\\
\hline
\end{tabular}
\end{center}
\end{table}

\begin{table}
\caption{
Average values of selected $G$ (the number of groups) and $L$ (the number of latent distributions) in groped heterogeneous mixture (GHM) modeling under count responses.
}
\label{tab:GHM-Po}
\begin{center}
\begin{tabular}{cccccccccccccccc}
\hline
&& \multicolumn{2}{c}{(I)} & \multicolumn{2}{c}{(II)} & \multicolumn{2}{c}{(III)} & \multicolumn{2}{c}{(IV)} & \multicolumn{2}{c}{(V)} \\
$(m, n)$ &  & $G$ & $L$ & $G$ & $L$ & $G$ & $L$ & $G$ & $L$ & $G$ & $L$\\
\hline
(20, 80) &  & 4.11 & 2.00 & 2.66 & 2.00 & 3.15 & 2.00 & 4.72 & 2.24 & 3.97 & 2.00\\
(20, 160) &  & 5.05 & 2.00 & 2.39 & 2.00 & 3.45 & 2.00 & 6.08 & 2.55 & 4.82 & 2.00\\
(40, 80) &  & 5.74 & 2.00 & 3.37 & 2.00 & 4.61 & 2.00 & 6.33 & 2.46 & 5.30 & 2.00\\
(40, 160) &  & 6.98 & 2.00 & 2.49 & 2.00 & 5.11 & 2.00 & 8.02 & 2.76 & 7.00 & 2.00\\
(60, 80) &  & 6.97 & 2.00 & 4.10 & 2.00 & 5.99 & 2.00 & 7.47 & 2.65 & 6.47 & 2.00\\
\hline
\end{tabular}
\end{center}
\end{table}

\subsection{Example: crime risk modeling in 23 wards of Tokyo}\label{sec:crime}
We here apply the proposed method to crime risk modeling in 23 wards of Tokyo, Japan.
We used a dataset of number of police-recorded crime in Tokyo metropolitan area, provided by University of Tsukuba and publicly available online (``GIS database of number of police-recorded crime at O-aza, chome in Tokyo, 2009-2017'', available at \url{https://commons.sk.tsukuba.ac.jp/data_en}). 
In this study, we focus on the number of violent crimes in $m=23$ wards in Tokyo metropolitan area in 2015, and wards are recognized as clusters. 
For $i=1,\ldots,m$, each ward contains $n_i$ local towns, and $n_i$ ranges from 51 to 275, and the number of total local towns (sample sizes) is 2855. 
The wards are indexed as $i=1,\ldots,m$ and local towns in the $i$th ward are indexed as $j=1,\ldots,n_i$. 
For auxiliary information for each town, area (km$^2$), population densities in noon and night, density of foreign people, percentage (\%) of unemployment, single-person household and university graduation and average duration of residence are available in each local town.

Let $y_{ij}$ be the number of crimes in the $j$th local town in the $i$th ward, and we found that the ratios of zero counts in $y_{ij}$ is relatively large.
For modeling the conditional distribution of $y_{ij}$ given the covariates, we employed the following mixture model of Poisson regressions and a zero-inflated component:
\begin{equation}\label{ZIP}
y_{ij}|x_{ij}\sim \pi_{g_i1}\delta_0(y_{ij})+\sum_{k=2}^L\pi_{g_ik}{\rm Po}(y_{ij}; a_{ij}\exp(x_{ij}^t\beta_k)), 
\end{equation} 
where $\delta_0(\cdot)$ denotes the one-point distribution on the origin which corresponds to the zero-inflation component, $a_{ij}$ is the area, $x_{ij}$ is a vector of 7 covariates described above and an intercept term, and $g_i\in \{1,\ldots,G\}$ are grouping parameters.  
Note that $\log a_{ij}$ can be regarded as an offset term, and $\exp(x_{ij}^t\beta_k)$ can be interpreted as the crime risk per square kilometer. 
In the model (\ref{ZIP}), the cluster-specific mixing proportions act as ward-specific effects to address the heterogeneity among the wards. 
Based on the adjacent information among 23 woads, we define the adjacent indicator $s_{ii'}$ for $i, i'\in \{1,\ldots,m\}$, that is, $s_{ii'}=1$ if the $i$th ward and the $i'$th ward are adjacent and $s_{ii'}=0$ otherwise, and adopted the penalized likelihood (\ref{PML}) with $\xi=0.2$ to address potential spatial similarity.
We also tried $\phi=0.3$ and $\phi=0.5$, but the following results did not change much. 
Based on the information criterion (\ref{BIC}), we estimated $G$ and $L$ from candidates $\{1,2,\ldots,8\}$ and $\{2,3,4\}$, respectively, and $G=4$ and $L=3$ were selected.
For comparison, we also applied the LDM model with $L=3$ in which the first component is the zero-inflation component and the other two components are Poisson regression models.

In Table \ref{tab:est}, we reported the estimates of regression coefficients in the latent two Poisson regression models, which reveals that the coefficients in GHM and LDM are not much different since the modeling frameworks of GHM and LDM are quite similar. 
Comparing the regression coefficients in the two components, they are quite different in some covariates such as population density in night and percentage of single-person household, which would indicate the importance of using mixture modeling in this case. 
To visualize the grouping structure, we present the geographical map of estimated groping assignment in 23 wards in the left panel of Figure \ref{fig:group}.
The figure shows that the 23 wards can be successfully divided into $G=4$ groups and adjacent structures could be incorporated into the grouping result. 
For comparing estimated mixing proportions in GHM and LDM, in the right panel of Figure \ref{fig:group}, we show estimates of cluster-wise (ward-wise) mixing proportions in LDM together with group-wise mixing proportions in GHM for the second and third components.
We colored the estimates of cluster-wise mixing proportions according to the grouping assignment in GHM. 
It is noted that the estimated prior probabilities in LDM are $0.14$ and $0.74$ for the second and third components, respectively, thereby the estimated mixing proportions in LMD are scattered around these values. 
From the figure, it is observed that the variations of the mixing proportions in LDM seems less than the group-wise mixing proportions in GHM, possibly because the use of the parametric Dirichlet distribution for the mixing proportions in LDM may restrict variations in the estimates. 
On the other hand, since the mixing proportions in the four groups are relatively different each other, we may easily interpret the characteristics of the groups, which could be one of practical advantages of the proposed method compared with LDM.   
We also note that the mixing proportions of the first component (zero inflated term) are $0.11$, $0.23$, $0.00$, $0.13$ in groups $1\sim 4$, respectively, so that the zero-inflation component would also be an meaningful component to distinguish the four groups.

We next investigated the predictive performance of GHM and LDM.
We randomly omitted $h\%$ of samples, where $h\in \{5, 10, 15, 20\}$, to act as test data, and set the rest of the data as training data. 
We then estimated GHM and LDM based on the training data, where $G$ and $L$ in GHM were selected via the information criterion (\ref{BIC}) from the same candidates and $L=3$ was adopted in LDM, and predict the response values (number of crimes) using the covariates and offset in the test data.
For measuring the performance of the prediction, we calculated the negative predictive log-likelihood (PLL) and mean squared error (MSE) given by 
\begin{align*}
&{\rm PLL} = -\sum_{i=1}^m\sum_{j\in r_i}\log\left\{\pih_{i1}\delta_0(y_{ij})+\sum_{k=2}^L\pih_{ik}{\rm Po}(y_{ij}; a_{ij}\exp(x_{ij}^t\beh_k))\right\}\\
&{\rm MSE} = \sum_{i=1}^m\sum_{j\in r_i}\Big\{\log(\muh_{ij}+1)-\log(y_{ij}+1)\Big\}^2, \ \ \ \ \ \muh_{ij}=\sum_{k=2}^L\pih_{ik}a_{ij}\exp(x_{ij}^t\beh_k),
\end{align*}  
where $r_i$ denotes the index set of the test data in the $i$th cluster.
Note that smaller values indicate better prediction performance in both criterion, and PLL and MSE are concerned with performance in terms of distribution prediction and point prediction, respectively. 
We repeated the procedure for 100 times, and reported the mean and median values in Table \ref{tab:pred}.
It shows that the proposed GHM approach outperforms the LDM method in terms of both distribution and point prediction, and the amount of improvement tends to increase as the percentage of test data increases.  
This may be because the characteristics of GHM that provides parsimonious modeling and is numerically stable compared with LDM results in good prediction especially when the training data used in the estimation is limited.

\begin{table}
\caption{Estimates of regression coefficients in GHM and LDM with the zero-inflated Poisson mixture models.
The coefficients other than intercept are scaled by the standard deviations of the corresponding covariates. 
\label{tab:est}
}
\begin{center}
\begin{tabular}{cccccccc}
\hline
 &  & \multicolumn{2}{c}{2nd component} && \multicolumn{2}{c}{3rd component} \\
Covariate & & GHM & LDM & & GHM & LDM\\
 \hline
Intercept &  & -5.10 & -5.27 &  & -6.95 & -6.89 \\
Population density in noon &  & 0.212 & 0.182 &  & 0.410 & 0.409 \\
Population density in night &  & 0.428 & 0.455 &  & 0.180 & 0.158 \\
Density of foreign people &  & 0.134 & 0.163 &  & 0.104 & 0.102 \\
Percentage of unemployment &  & -0.151 & -0.132 &  & -0.088 & -0.073 \\
Percentage of single-person household &  & 0.244 & 0.228 &  & 0.440 & 0.459 \\
Percentage of university graduation &  & 0.101 & 0.104 &  & 0.175 & 0.182 \\
Average duration of residence &  & -0.103 & -0.093 &  & -0.062 & -0.052 \\
\hline
\end{tabular}
\end{center}
\end{table}

\begin{table}
\caption{Mean and median values of 100 replicated values of negative predictive log-likelihood (PLL) and mean squared error (MSE) that measures the prediction performance of GHM and LDM under $h\%$ of samples are set to test data.
\label{tab:pred}
}
\begin{center}
\begin{tabular}{cccccccccccc}
\hline
 &  & \multicolumn{2}{c}{PLL (mean)} & \multicolumn{2}{c}{PLL (median)} && \multicolumn{2}{c}{MSE (mean)} & \multicolumn{2}{c}{MSE (median)}\\
$h$ & & GHM & LDM & GHM & LDM && GHM & LDM & GHM & LDM \\
\hline
5 &  & 210.1 & 213.4 & 209.2 & 212.5 &  & 54.1 & 62.7 & 50.9 & 51.2 \\
10 &  & 417.4 & 424.9 & 417.5 & 423.6 &  & 120.5 & 143.1 & 104.8 & 107.4 \\
15 &  & 625.9 & 635.0 & 623.6 & 634.0 &  & 166.3 & 186.0 & 155.0 & 161.7 \\
20 &  & 834.3 & 844.0 & 834.3 & 843.9 &  & 231.2 & 264.6 & 213.1 & 253.3 \\
\hline
\end{tabular}
\end{center}
\end{table}

\begin{figure}[!htb]
\centering
\includegraphics[bb=0 0 1000 1000, width=7cm]{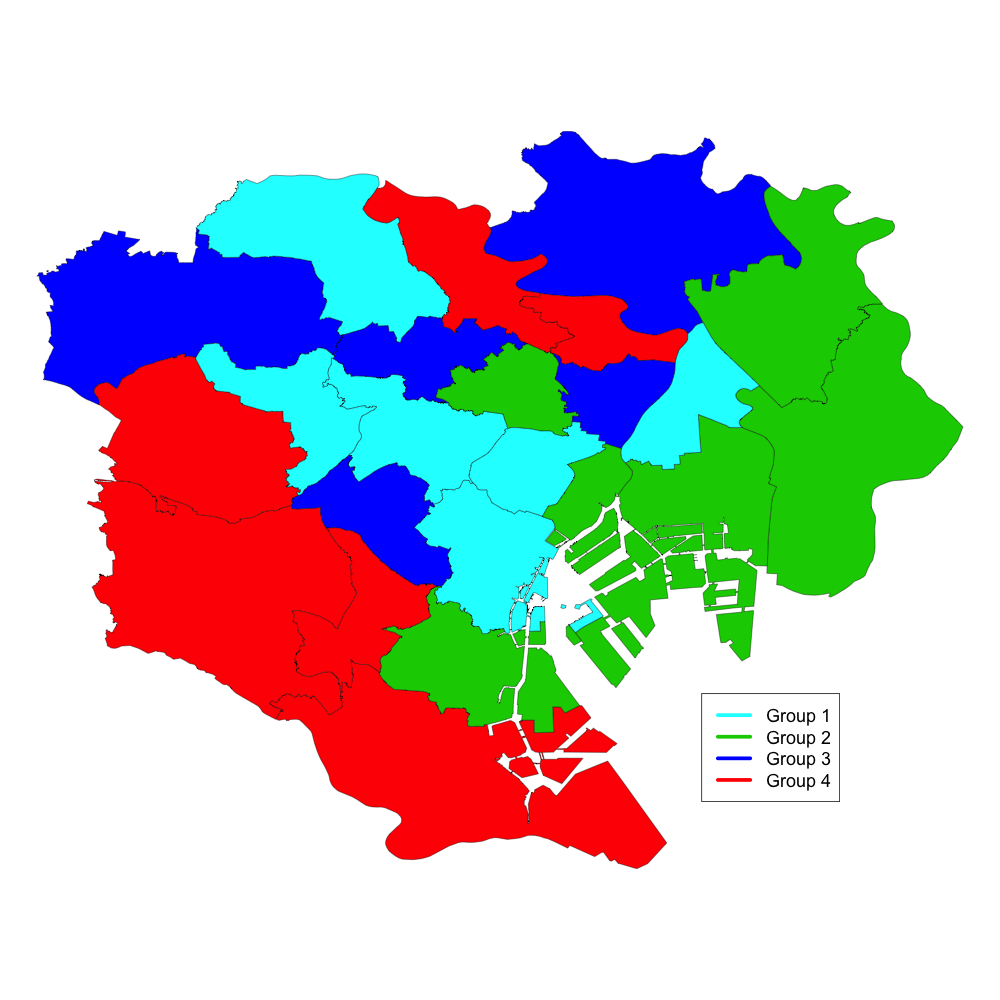}
\includegraphics[width=7cm]{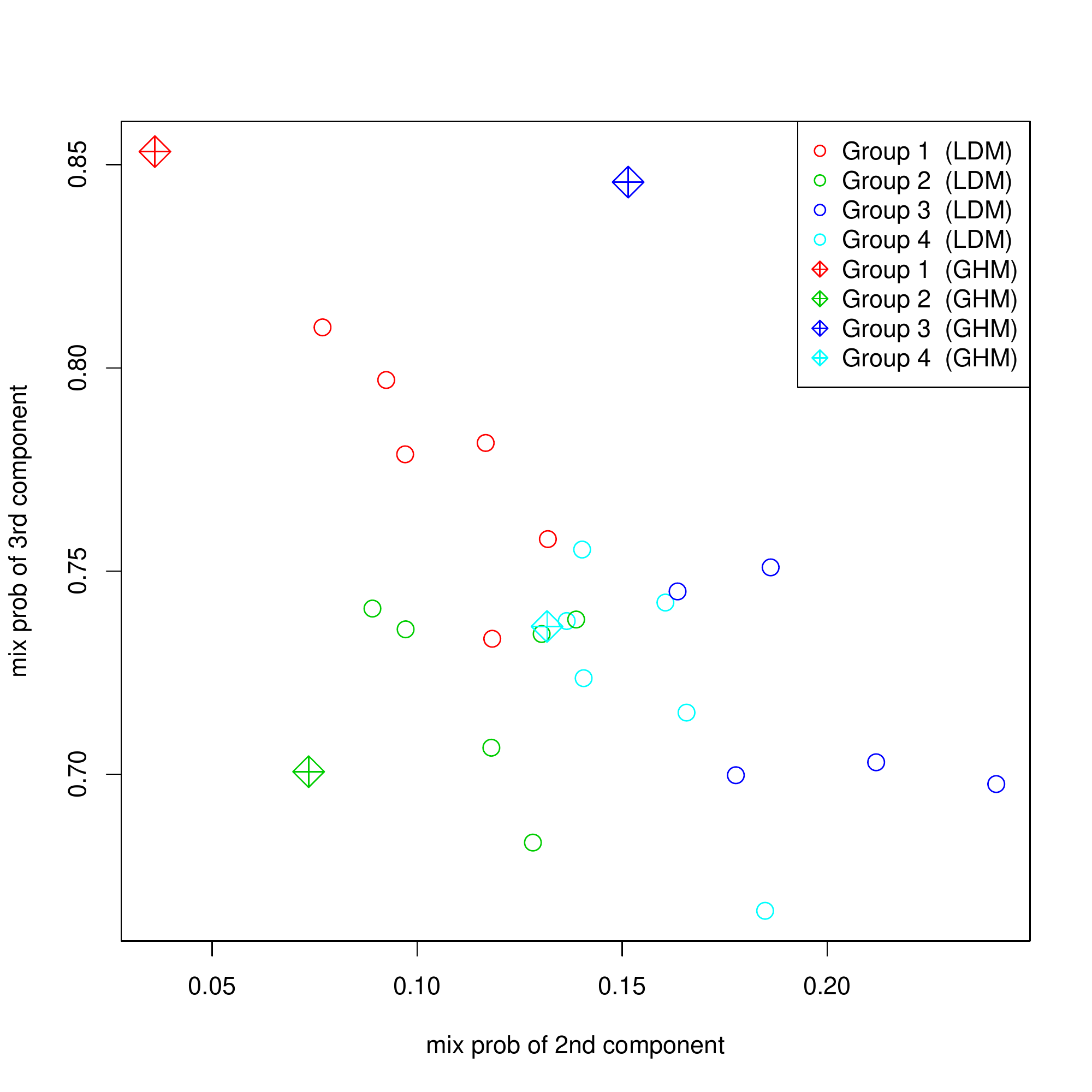}
\caption{Estimated grouping of 23 wards in Tokyo (left) and estimates of cluster-wise mixing proportions in LDM and group-specific mixing proportions in GHM (right).}
\label{fig:group}
\end{figure}

\section{Discussion}\label{sec:disc}
We proposed a new strategy for flexible modeling of conditional distributions for clustered data, called grouped heterogeneous mixture (GHM) modeling, while accounting for heterogeneity among clusters.
The key of the proposed modeling is that cluster-wise conditional distributions are expressed as finite mixtures of latent distributions which are common to all the clusters and between-cluster heterogeneity is expressed via cluster-wise mixing proportions that has a grouping structure. 
We developed the generalized EM algorithm for estimating model parameters as well as unknown grouping parameters, which can be easily carried out owing to the simplicity of each step in the algorithm. 
We proposed the information criterion to select the numbers of groups and latent distributions.
We also considered structured grouping via penalization to incorporate auxiliary information into grouping assignment.
Under the setting where both the number of clusters and cluster sizes tend to infinity, we provided asymptotic properties of the maximum likelihood estimator and selection consistency of the information criterion.
From the results in our numerical studies, we confirmed that the proposed method is quite promising for modeling clustered data.

As considered in Section \ref{sec:pen}, the proposed method may have potential extensions to address important practical issues. 
For example, when the number of covariates or dimension of $x_{ij}$ is large, it would be reasonable to select the subset of suitable covariates to achieve better interpretation of the model or better prediction accuracy. 
In the context of mixture modeling for independent data without cluster information, some regularization techniques have been considered \citep[e.g.][]{Hui2015, KC2007}, and these strategies can be readily incorporated into the proposed model.
When some penalty functions for the regression coefficients, which is a sub-vector of $\phi_k$, are introduced, the M-steps for model parameters other than $\phi_k$ will be unchanged, and the updating processes for $\phi_k$ will be the same as the maximization of the penalized weighted log-likelihood function of the latent distributions. 
The detailed discussion could be an important future work.

Finally, the proposed GHM modeling imposes `hard grouping' of clusters in the sense that ambiguity of grouping assignment is completely ignored. 
While this structure provides us easy interpretations of the estimation results, we may lose a certain degree of flexibility when the data does not admit any grouping structures. 
To overcome the difficulty, we may consider `soft grouping' version of the proposed method by considering pseudo-probability of grouping assignment instead of deterministic grouping in the estimation algorithm.
This direction of research could also be an interesting future work.

\section*{Acknowledgement}
This work is partially supported by Japan Society for Promotion of Science (KAKENHI) grant numbers 16H07406 and 18K12757.

\vspace{0.5cm}
\appendix 

\setcounter{equation}{0}
\renewcommand{\theequation}{A\arabic{equation}}
\setcounter{section}{0}
\renewcommand{\thesection}{A\arabic{section}}
\setcounter{lem}{0}
\renewcommand{\thelem}{A\arabic{lem}}

\begin{center}
{\bf\Large Appendix}
\end{center}

\section{Proof of Theorem \ref{thm:main}}

We require the following regularity conditions:

\begin{itemize}
\item[(C1)]
The density (probability) function $g(Y|X; \gamma,\psi)\equiv \prod_{i=1}^m\prod_{j=1}^{n_i}f_i(y_{ij}|x_{ij};  g_i,\psi)$ is identifiable in $\Theta$ up to the permutation of the component and grouping labels.

\item[(C2)]
The Fisher information matrix
$$
F_m(\Theta)=\E\left[\left(\frac{\partial}{\partial\psi}\log g(Y|X; \gamma,\psi)\right)\left(\frac{\partial}{\partial\psi}\log g(Y|X; \gamma,\psi)\right)^t\right]
$$
is finite and positive definite at $\psi=\psi^0$ and $\gamma=\gamma^0$.

\item[(C3)]
There exists an open subset $\omega$ of $\Omega$ containing true parameters $\psi$, such that there exists functions $M_k(x,y), k=1,2,3,$ with
\begin{align*}
&\frac{f_i(y|x; \gamma,\psi)}{f_i(y|x; \gamma^{\dagger},\psi)}<M_1(x,y),\ \ \ \ \ 
\Big|\frac{\partial}{\partial\psi_s}\log f_i(y|x; \gamma,\psi)\Big|<M_{2s}(x,y), \\ 
&\Big|\frac{\partial^2}{\partial\psi_r\partial\psi_s}\log f_i(y|x; \gamma,\psi)\Big|<M_{3rs}(x,y),
\end{align*}
for arbitrary $\gamma$ and $\gamma^{\dagger}$.
Moreover,  $\E[M_1(x,y)^2]<\infty$, $\E[M_{2s}(x,y)^2]<\infty$ and $\E[M_{3rs}(x,y)^2]<\infty$.

\item[(C4)]
$\psi^0$ is an interior point in the compact set $\Omega\subset \Re^{{\rm dim}(\psi)}$.

\end{itemize}

\noindent
We define a set of neighborhood of $\psi^0$ as $N_\eta=\{\psi^0+\eta u; \|u\|=1\}$, where $\eta$ is a scalar, $u\in\Re^{{\rm dim}(\psi)}$, and $\|\cdot\|$ denotes the $L^2$-norm.
We first show the following lemma.

\begin{lem}\label{lem:mis}
For any $\delta>0$ and $\psi\in N_\eta$, it holds that 
$$
\frac1m\sum_{i=1}^m\Ind(\gh_i(\psi)\neq g_{i0})=o_p(n_1^{-\delta}).
$$
\end{lem}

\begin{proof}
From Markov's inequality, it follows that 
\begin{equation}\label{misprob}
\P\left(\frac1m\sum_{i=1}^m\Ind(\gh_i(\psi)\neq g_{i0})>\ep n_1^{-\delta}\right)
\leq \frac{n_1^{\delta}}{m\ep}\sum_{i=1}^m\P(\gh_i(\psi)\neq g_{i0})
\end{equation}
for arbitrary $\ep>0$ and $\psi\in N_\eta$.
Note that 
\begin{align*}
\P\big(\gh_i(\psi)\neq g_i^0\big)
&\leq \sum_{g\neq g_i^0}\P\Big(\sum_{j=1}^{n_i}\log f_i(y_{ij}; g,\psi)\geq \sum_{j=1}^{n_i}\log f_i(y_{ij}; g_i^0,\psi)\Big)\\
&=\sum_{g\neq g_i^0}\P\Big(\sum_{j=1}^{n_i}\log \frac{f_i(y_{ij}; g,\psi)}{f_i(y_{ij}; g_i^0,\psi)}\geq 0\Big).
\end{align*}
Moreover, from Markov's inequality, it holds that 
\begin{align*}
\P\Big(\sum_{j=1}^{n_i}\log \frac{f_i(y_{ij}; g,\psi)}{f_i(y_{ij}; g_i^0,\psi)}\geq 0\Big)
&\leq \left\{\E\left[\exp\left(\frac12\log \frac{f_i(y_{ij}; g,\psi)}{f_i(y_{ij}; g_i^0,\psi)}\right)\right]\right\}^{n_i}\\
&=\exp\left\{n_i\log\E\left[\sqrt{\frac{f_i(y_{ij}; g,\psi)}{f_i(y_{ij}; g_i^0,\psi)}}\right]\right\}.
\end{align*}
Note that 
\begin{align*}
\sqrt{\frac{1}{f_i(y_{ij}; g_i^0,\psi)}}
&=\sqrt{\frac{1}{f_i(y_{ij}; g_i^0,\psi^0)}}+f(y_{ij}; g_i^0,\psi^{\dagger})^{-3/2}(\psi-\psi^0)^t\frac{\partial}{\partial\psi}f_i(y_{ij}; g_i^0,\psi)\Big|_{\psi=\psi^{\dagger}}\\
&=\sqrt{\frac{1}{f_i(y_{ij}; g_i^0,\psi^0)}}+f(y_{ij}; g_i^0,\psi^{\dagger})^{-1/2}(\psi-\psi^0)^t\frac{\partial}{\partial\psi}\log f_i(y_{ij}; g_i^0,\psi)\Big|_{\psi=\psi^{\dagger}},
\end{align*}
where $\psi^{\dagger}$ lies on the line segment joining $\psi$ and $\psi^0$.
Therefore, under Assumption (C3), we have
\begin{align*}
&\E\left[\sqrt{\frac{f_i(y_{ij}; g,\psi)}{f_i(y_{ij}; g_i^0,\psi)}}\right]=\int \sqrt{\frac{f_i(y_{ij}; g,\psi)}{f_i(y_{ij}; g_i^0,\psi)}}f_i(y_{ij}; g_i^0,\psi^0){\rm d}y_{ij}\\
&=\int \sqrt{f_i(y_{ij}; g,\psi)f_i(y_{ij}; g_i^0,\psi^0)}{\rm d}y_{ij}+(\psi-\psi^0)^t \E\left[\left(\frac{\partial}{\partial\psi}\log f_i(y_{ij}; g_i^0,\psi)\Big|_{\psi=\psi^{\dagger}}\right)\sqrt{\frac{f_i(y_{ij}; g,\psi)}{f(y_{ij}; g_i^0,\psi^{\dagger})}}\right]\\
&\leq 1-H(f(\cdot; g,\psi),f(\cdot; g_i^0,\psi^0))+C\eta,
\end{align*}
where $H(f_0,f_1)$ denotes the Hellinger distance between $f_0$ and $f_1$ and $C$ is a constant. 
Under Assumption (C1), $\inf_{\psi\in N_\eta}H(f(\cdot; g,\psi),f(\cdot; g_i^0,\psi^0))>0$ when $g\neq g_i^0$, so that for sufficiently small $\eta$, there exist a constant $c>0$ such that
$$
\sup_{\psi\in N_\eta}\log\Big\{1-H(f(\cdot; g,\psi),f(\cdot; g_i^0,\psi^0))+C\eta\Big\}<-c.
$$
Hence, we have
$$
\P\Big(\sum_{j=1}^{n_i}\log \frac{f_i(y_{ij}; g,\psi)}{f_i(y_{ij}; g_i^0,\psi)}\geq 0\Big)\leq \exp(-cn_i)=o(n_1^{-\delta}),
$$
thereby the right side of (\ref{misprob}) is $o(1)$, which completes the proof. 
\end{proof}

\vspace{0.3cm}\noindent
We next show the consistency of $\Thh$ as given in the following lemma:
\begin{lem}\label{lem:cons}
As $m\to\infty$ and $n_1\to\infty$, it holds that $\Thh\to\Theta^0$.
\end{lem}

\begin{proof}
Define $u\in\Re^{{\rm dim}(\psi)}$ such that $\|u\|=1$ and $\gamma_s^{\dagger}=(g_{1s},\ldots,g_{ms})\in \{1,\ldots,G\}^m$ such that $d(\gamma_s^{\dagger},\gamma^0)\equiv \sum_{i=1}^m \Ind(g_{is}^{\dagger}\neq g_i^0)=s$.
Further, we define 
$$
R(\psi,\gamma)\equiv \frac1N\sum_{i=1}^mR_i(\psi,g_i)=\frac1NQ(\Theta).
$$
Note that 
\begin{align*}
&|R(\psi^0+\eta u,\gamma_s^{\dagger})-R(\psi^0,\gamma^0)|\\
& \ \ \ \ \ \ 
=|R(\psi^0+\eta u,\gamma^0)-R(\psi^0,\gamma^0)|+|R(\psi^0+\eta u,\gamma_s^{\dagger})-R(\psi^0+\eta u,\gamma^0)|
=I_1+I_2.
\end{align*}
Under Assumption (C2) and (C3), using the similar argument given in the proof of Theorem 1 in \cite{Hui2015}, it holds that $I_1=o(\eta)$.
Regarding $I_2$, it is noted that 
\begin{align*}
I_2&\leq \frac1N\sum_{i=1}^m|R_i(\psi^0+\eta u,g_{is}^{\dagger})-R_i(\psi^0+\eta u,g_i^0)|\\
&=\frac1N\sum_{i=1}^mI(g_{is}^{\dagger}\neq g_i^0)|R_i(\psi^0+\eta u,g_{is}^{\dagger})-R_i(\psi^0+\eta u,g_i^0)|\leq sC,
\end{align*}
for some $C>0$.
Therefore, for any $s$, there exists a local maximum inside the set $\{\psi^0+\eta u; \|u\|=1\}\bigcup \{\gamma^{\dagger}; d(\gamma^{\dagger},\gamma^0)<s\}$, so that the consistency follows.
\end{proof}

\vspace{0.3cm}\noindent
We define the following two objective functions:
\begin{equation}\label{Q2}
\begin{split}
&\Qh(\psi)=\sum_{i=1}^m\sum_{j=1}^{n_i}\log \ft_{ij}(\psi) \equiv \sum_{i=1}^m\sum_{j=1}^{n_i}\log\left(\sum_{k=1}^L\pi_{\gh_i(\psi)k}h_k(y_{ij}|x_{ij};\phi_k)\right), \\
&\Qt(\psi)=\sum_{i=1}^m\sum_{j=1}^{n_i}\log f_{ij}^0(\psi) \equiv \sum_{i=1}^m\sum_{j=1}^{n_i}\log\left(\sum_{k=1}^L\pi_{g_i^0k}h_k(y_{ij}|x_{ij};\phi_k)\right),
\end{split}
\end{equation}
and we define $\psih$ and $\psit$ as the maximizer of $\Qh(\psi)$ and $\Qt(\psi)$, respectively.
Note that $\sqrt{N}F_m(\psit)^{1/2}(\psit-\psi^0)\to N(0,I_{{\rm dim}(\psi)})$.
It holds that 
\begin{align*}
\bigg|\frac1N\Qh(\psi)-\frac1N\Qt(\psi)\bigg|
&\leq \frac1N\sum_{i=1}^m\sum_{j=1}^{n_i}I(\gh_i(\psi)\neq g_i^0)\big|\log \ft_{ij}(\psi)-\log f_{ij}^0(\psi)\big|\\
&\leq \left(\frac1N\sum_{i=1}^m n_i I(\gh_i(\psi)\neq g_i^0)\right)^{\frac12}\left\{\frac1N\sum_{i=1}^m\sum_{j=1}^{n_i}\left(\log \frac{\ft_{ij}(\psi)}{f_{ij}^0(\psi)}\right)^2\right\}^{\frac12}\\
&\leq \left(\frac{n_m}{mn_1}\sum_{i=1}^m I(\gh_i(\psi)\neq g_i^0)\right)^{\frac12}\left\{\frac1N\sum_{i=1}^m\sum_{j=1}^{n_i}\left(\log \frac{\ft_{ij}(\psi)}{f_{ij}^0(\psi)}\right)^2\right\}^{\frac12},
\end{align*}
thereby, from Lemma \ref{lem:mis}, we have 
\begin{equation}\label{sup}
\sup_{\psi\in N_\eta}|N^{-1}\Qh(\psi)-N^{-1}\Qt(\psi)|=o_p(n_1^{-(\delta+1-\alpha)/2}).
\end{equation}
Then, from the definition of $\psit$ and $\psih$, we have
\begin{align*}
0&\leq \Qt(\psit)-\Qt(\psih)
=\{\Qt(\psit)-\Qh(\psit)\}+\{\Qh(\psit)-\Qt(\psih)\}\\
&\leq 
\{\Qt(\psit)-\Qh(\psit)\}+\{\Qh(\psih)-\Qt(\psih)\}=o_p(n_1^{-(\delta+1-\alpha)/2}),
\end{align*}
where the last equality follows from (\ref{sup}) and $\psih,\psit\in N_\eta$ for large $n_i$ since $\psih$ and $\psit$ are consistent.
from (\ref{sup}), so that $\psih-\psit=o_p(n_1^{-(\delta+1-\alpha)/2})$.
Hence, it follows that 
\begin{align*}
\sqrt{N}(\psih-\psi^0)&=\sqrt{N}(\psit-\psi^0)+\sqrt{N}(\psih-\psit)
=\sqrt{N}(\psit-\psi^0)+o_p(\sqrt{N}n_1^{-(\delta+1-\alpha)/2})\\
&=\sqrt{N}(\psit-\psi^0)+o_p(1),
\end{align*}
since $\sqrt{N}n_1^{-(\delta+1-\alpha)/2}\leq \sqrt{m}n_1^{-(\delta+1-2\alpha)/2}=\sqrt{m}n_1^{-(\nu+1)/2}$ by choosing $\delta=2\alpha+\nu$, and is $o(1)$ under $m/n^{\nu}\to 0$, which establishes (\ref{st2}).
Moreover, we have 
\begin{align*}
\frac1m&\sum_{i=1}^m\sum_{k=1}^L(\pih_{\gh_ik}-\pi^0_{g_i^0k})^2\\
&=\frac1m\sum_{i=1}^m\sum_{k=1}^L\left\{(\pih_{\gh_ik}-\pih_{g_i^0k})+(\pih_{g_i^0k}-\pit_{g_i^0k})+(\pit_{g_i^0k}-\pi^0_{g_i^0k})\right\}^2\\
&\leq 
\frac3m\sum_{i=1}^m\sum_{k=1}^L(\pih_{\gh_ik}-\pih_{g_i^0k})^2
+\frac3m\sum_{i=1}^m\sum_{k=1}^L(\pih_{g_i^0k}-\pit_{g_i^0k})^2
+\frac3m\sum_{i=1}^m\sum_{k=1}^L(\pit_{g_i^0k}-\pi^0_{g_i^0k})^2\\
&= o_p(n_1^{-\delta}) + o_p(n_1^{-\delta}) + O_p(m^{-1}),
\end{align*}
from Lemma \ref{lem:mis} and the same argument used in (\ref{sup}).
Hence, (\ref{st1}) follows under $m/n_1^{\nu}\to 0$.

\section{Proof of Theorem \ref{thm:selection}}
In this proof, let $\Theta^0$ be the true parameter. 
Moreover, we use $\Qh_{G, L}$ and $\Qt_{G, L}$ to denote $\Qh$ and $\Qt$ in (\ref{Q2}) to address the dependency on $G$ and $L$.
Here we require an additional condition given below.

\begin{itemize}
\item[(C5)]
There exists a constant $R$ such that $E[Q_o(\Theta^0)-Q_o(\Theta^{\ast}_{G, L})]\geq R >0$ for all models with $G<G_0$ or $L<L^0$, where $Q_o(\Theta_{G, L})$ denotes the log-likelihood of a single observation, and $\Theta^{\ast}_{G, L}$ denotes the pseudo-true parameters \citep[e.g.][]{White} which maximizes $E[Q_o(\Theta_{G, L})]$ with respect to $\Theta_{G, L}$.
\end{itemize}

Note that when $G<G_0$ or $L<L^0$, the model is under-specified, thereby the above condition requires that the Kullback-Leibler distance between any under-specified models with pseudo-true parameter and the true model with $G=G^0$ and $L=L^0$ is positive.
Conditions similar to (C5) are typically imposed in theoretical development on consistency of model selection or variable selection \citep[e.g.][]{Keribin2000, Hui2017b}.

We show that $P\{{\rm IC}(G, L)<{\rm IC}(G^0, L^0)\}=o(1)$ for each $(G, L)\neq (G^0, L^0)$.
We evaluate the probability separately in two cases, namely, $G<G^0$ or $L<L^0$ and $G>G^0$ and $L>L^0$.

We first consider the under-specified case.
It holds that 
\begin{equation}\label{dif1}
\frac1N{\rm IC}(G, L)-\frac1N{\rm IC}(G^0, L^0)
=\frac2NQ(\Thh_{G^0, L^0})-\frac2NQ(\Thh_{G, L})+\frac{P\log N}{N},
\end{equation}
where $P=G(L-1)-G^0(L^0-1)-\sum_{k=\min\{L,L^0\}+1}^{\max\{L,L^0\}}{\rm dim}(\phi_k)$ is the difference of the numbers of parameters, which is $O(1)$.
From the definition of $\Qh$ and $\Qt$ in (\ref{Q2}) and the property given in (\ref{sup}), for sufficiently large $N$, it holds that 
\begin{equation}\label{Q3}
\begin{split}
\frac1NQ(\Thh_{G^0, L^0})
&=\frac1N\Qh_{G^0, L^0}(\psih)=\frac1N\Qt_{G^0, L^0}(\psih)+o_p(n_1^{-\delta^{\ast}})\\
&=\frac1N\Qt_{G^0, L^0}(\psi^0)+o_p(n_1^{-\delta^{\ast}})=\frac1NQ(\Theta^0)+o_p(n_1^{-\delta^{\ast}}),
\end{split}
\end{equation}
for sufficiently large $\delta^{\ast}$.
For the under-specified choice, we can use a similar argument to the proof of Theorem \ref{thm:main} to show that $\Qh_{G, L}$ and $\Qt_{G, L}$ hold the same property as (\ref{sup}), so that a similar argument given above leads to $N^{-1}Q(\Thh_{G, L})=N^{-1}Q(\Theta^{\ast}_{G, L})+o_p(1)$, where $\Theta^{\ast}_{G, L}$ is the pseudo-true parameter.
Moreover, by the weak law of large numbers, we have $N^{-1}Q(\Theta^0)\to E[Q_o(\Theta^0)]$ and $Q(\Theta^{\ast}_{G, L})\to E[Q_o(\Theta^{\ast})]$ in probability. 
Therefore, it follow that 
\begin{align*}
\frac1N\left\{Q(\Thh_{G^0, L^0})-Q(\Thh_{G, L})\right\}
&=E[Q_o(\Theta^0)]-E[Q_o(\Theta^{\ast})]+o_p(1) \geq R+o_p(1).
\end{align*}
Then, (\ref{dif1}) is strictly positive with probability tending to $1$.

We next consider the over-specified case, namely, $G>G^0$ and $L>L^0$.
Again, we let $P=G(L-1)-G^0(L^0-1)-\sum_{k=L^0+1}^{L}{\rm dim}(\phi_k)$ be the difference of the numbers of parameters, which is $O(1)$ and $P>0$.
Note that   
$$
{\rm IC}(G, L)-{\rm IC}(G^0, L^0)=2\left\{Q(\Thh_{G^0,L^0})-Q(\Thh_{G,L})\right\}+P\log N,
$$
and 
\begin{align*}
&Q(\Thh_{G^0,L^0})-Q(\Thh_{G,L})\\
&=\left\{Q(\Thh_{G^0,L^0})-Q(\Theta^0)\right\}-\left\{Q(\Thh_{G,L^0})-Q(\Theta^0)\right\}+\left\{Q(\Thh_{G,L^0})-Q(\Thh_{G,L})\right\}\\
&\geq -\left\{Q(\Thh_{G,L^0})-Q(\Theta^0)\right\}+\left\{Q(\Thh_{G,L^0})-Q(\Thh_{G,L})\right\}.
\end{align*}
From a similar argument to (\ref{Q3}), we can obtain that the first term is $o_p(n_1^{-\delta^{\ast}})$ for sufficiently large $\delta^{\ast}$.
Regarding the second term, it holds that
\begin{align*}
&Q(\Thh_{G,L})-Q(\Thh_{G,L^0})\\
&=\sum_{i=1}^m\sum_{j=1}^{n_i}I(\gh_i(L)=\gh_i(L^0))\log \frac{f_{ij}(\Thh_{G,L})}{f_{ij}(\Thh_{G,L^0})}
+\sum_{i=1}^m\sum_{j=1}^{n_i}I(\gh_i(L)\neq \gh_i(L^0))\log \frac{f_{ij}(\Thh_{G,L})}{f_{ij}(\Thh_{G,L^0})}\\
&\equiv I_1 + I_2.
\end{align*}
From a similar argument to Lemma \ref{lem:mis} and the derivation of (\ref{sup}), it follows that $N^{-1}I_2=o_p(n_1^{-\delta^{\ast}})$.
On the other hand, we have 
\begin{align*}
I_1=\sum_{g=1}^G\sum_{i=1}^m\sum_{j=1}^{n_i}I\{\gh_i(L)=\gh_i(L^0)=g\}\log \frac{f_{ij}(\Thh_{G,L})}{f_{ij}(\Thh_{G,L^0})}=\sum_{g=1}^G\ell_g,
\end{align*}
where $\ell_{g}$ is the likelihood ratio statistic for testing the number of components $L_0$ against $L$ in the mixture model fitted only to the data classified $g$th group. 
From (C4) and \cite{DG1999}, $\ell_{g}$ converges to a random variable characterized by the supremum of a stochastic process, thereby we obtain $I_1=O_p(1)$. 
Finally, we have 
$$
{\rm IC}(G, L)-{\rm IC}(G^0, L^0)=-2\times O_p(1)+o_p(1)+P\log N,
$$
so that the right hand side tends to infinity as $N\to\infty$, and the difference is guaranteed to be positive with probability tending to $1$.

\vspace{1cm}
\bibliographystyle{chicago}
\bibliography{ref}

\end{document}